\documentclass[a4paper,USenglish,cleveref, autoref, thm-restate,nolineno]{lipics-v2021}

\title{Coloring for dispersion: A polynomial-time algorithm for cardinality-constrained 2-anticlustering} 

\author{Nguyen Khoa Tran\footnote{Corresponding author}}{Department of Computer Science, Heinrich Heine University D\"usseldorf, Germany \and Center for Digital Medicine, D\"usseldorf, Germany}{nguyen.tran@hhu.de}{https://orcid.org/0000-0002-4732-4294}{}

\author{Lin Mu}{Department of Computer Science, Heinrich Heine University D\"usseldorf, Germany}{lin.mu@hhu.de}{}{}

\author{Martin Papenberg}{Department of Experimental Psychology, Heinrich Heine University D\"usseldorf, Germany \and Center for Digital Medicine, D\"usseldorf, Germany}{martin.papenberg@hhu.de}{https://orcid.org/0000-0002-9900-4268}{}

\author{Gunnar W. Klau}{Department of Computer Science, Heinrich Heine University D\"usseldorf, Germany \and Center for Digital Medicine, D\"usseldorf, Germany}{gunnar.klau@hhu.de}{https://orcid.org/0000-0002-6340-0090}{}

\authorrunning{N.\,K. Tran et al.} 

\Copyright{Nguyen Khoa Tran, Lin Mu, Martin Papenberg, and Gunnar W. Klau} 



\category{} 

\relatedversion{} 

\acknowledgements{We would like to thank Michael Brusco for raising the question that motivated this
  work in his signed review of \cite{papenbergExtendingBicriterionApproach2025a}. Furthermore, we would like to thank Philipp Spohr, Sara K\"aufler, Jonathan Bobak, and Melanie Schmidt's group for insightful discussions.}

\nolinenumbers 

\usepackage{graphicx}
\usepackage{subcaption}
\usepackage{amssymb}
\usepackage{amsmath}
\usepackage{xcolor}
\usepackage{algorithm}
\usepackage{algorithmicx}
\usepackage{algpseudocode}
\algtext*{EndFor}
\algtext*{EndWhile}
\algtext*{EndIf}
\algtext*{EndProcedure}

\DeclareMathOperator*{\argmax}{arg\,max}
\DeclareMathOperator*{\abs}{abs}

\newtheoremstyle{problemstyle}  
        {3pt}                                               
        {3pt}                                               
        {\normalfont}                               
        {}                                                  
        {\bfseries}                 
        {\normalfont\bfseries:}         
        {.5em}                                          
        {}                                                  
\theoremstyle{problemstyle}

\newtheorem{ntheorem}{Theorem}

\begin{document}

\maketitle

\begin{abstract}
The $k$-Maximum Dispersion Problem with Cardinality Constraints ($k$-MDCC) asks for a partition of a given item set with pairwise dissimilarities into $k$ cardinality-constrained groups such that the minimum pairwise intra-group dissimilarity, which is also known as the dispersion, is maximized. The problem arises in the context of anticlustering, where the goal is to create maximally heterogeneous groups of items with applications in psychological research, bioinformatics, and data science.

It is known that $k$-MDCC is NP-hard for $k \geq 3$ but it has been an open question whether it can be solved in polynomial time for $k = 2$. We give a positive answer to this question by showing that $2$-MDCC can be solved by a quadratic number of cardinality-constrained 2-coloring problem instances ($2$-COLCC).  We solve these instances by transforming them into a restricted class of subset sum instances. Although subset sum is NP-complete in general, for this restricted class the input values are bounded, ensuring that the pseudopolynomial dynamic programming algorithm runs in polynomial time. As a consequence, we obtain a polynomial-time algorithm for $2$-MDCC.

We demonstrate that a publicly available open-source implementation of our new algorithm outperforms the previous integer linear programming solution by several orders of magnitude so that even large datasets ($n = 10{,}000$) can be processed in less than a second.
\end{abstract}

\section{Introduction}\label{sec1}
Contrary to in its better known inverse field of clustering, the goal of \emph{anticlustering} is to partition a set of items into groups such that items within the same group are as dissimilar as possible, while the groups themselves remain as similar as possible \cite{Valev1983,Spaeh1986}. Practical applications arise, for example, in experimental psychology when assigning subjects to experimental and control groups, in the generation of exams of equivalent characteristics, in sample allocation to minimize batch effects, or in creating similar training splits in supervised machine learning. 

A common objective in anticlustering is to maximize dispersion, which is defined as the minimum pairwise dissimilarity between items within the same group. Likewise, many applications require cardinality restrictions on the groups. Here, we therefore focus on $2$-MDCC, an important special case of the $k$-Maximum Dispersion Problem with Cardinality Constraints ($k$-MDCC)~\cite{FERNANDEZ2013721,Brusco2020,papenbergExtendingBicriterionApproach2025a}, which we formalize as follows:

\noindent\rule{\textwidth}{0.4pt}
\textsc{$k$-Maximum Dispersion Problem With Cardinality Constraints ($k$-MDCC)}
\begin{itemize}
\item \textbf{Input:}
  A set of items $I$ consisting of $n$ datapoints with $m$ features, a dissimilarity function $d: \mathbb{R}^m \times \mathbb{R}^m \rightarrow \mathbb{R}$, the number of groups $k$, and cardinality constraints $c_1, c_2, \ldots, c_k$ with $\sum_{g=1}^{k} c_g = n$
\item \textbf{Output:} A partition $\pi^*: I \rightarrow \{1,\ldots, k\}$ with
\begin{align*}
    \pi^* = \argmax_\pi \;& \min_{g \in \{1,\ldots,k\}} \min_{i,j \in I: \pi(i) = \pi(j) = g} d(i, j) & \text{(maximize dispersion)}\\
\text{subject to } \;& |\{i \mid i \in I: \pi(i)=g\}| = c_g \ & \forall g \in \{1,\ldots,k\}
\end{align*}
\end{itemize}\vspace*{-3ex}
\noindent\rule{\textwidth}{0.4pt}

Fernández et al.~\cite{FERNANDEZ2013721} showed that $k$-MDCC is NP-hard under balanced cardinality constraints, which is a special case of the general $k$-MDCC problem with arbitrary cardinality constraints. This implies that the general problem is also NP-hard. 
However, their result does not explicitly address the complexity for the case $k = 2$, an open question that has been raised by a reviewer of~\cite{papenbergExtendingBicriterionApproach2025a}.

The \textsc{OptDispF} algorithm~\cite{papenbergExtendingBicriterionApproach2025a} uses the Euclidean distance as the dissimilarity function~$d$ and solves $k$-MDCC by iterating over all pairwise distances in ascending order. In each iteration, the problem is reduced to a $k$-coloring problem with cardinality constraints, while all remaining steps are solvable in polynomial time. The case $k = 2$ is nontrivial: while classical $2$-coloring can be solved in polynomial time using a breadth-first search (BFS), it is unclear whether this remains true once cardinality constraints are imposed. We formalize the corresponding decision problem as follows:

\noindent\rule{\textwidth}{0.4pt}
\textsc{2-Coloring Problem With Cardinality Constraints (2-COLCC)}
\begin{itemize}
\item \textbf{Input:} A graph $G = (V,E)$ and cardinality constraints $c_1, c_2$ with $c_1 + c_2 = |V|$
\item \textbf{Question:} Does there exist a coloring $c: V \rightarrow \{1,2\}$ such that:
\begin{enumerate}
    \item $c(u) \neq c(v) \ \forall \{u, v\} \in E$, and
    \item $|\{v \in V \mid c(v) = 1\}| = c_1$ and $|\{v \in V \mid c(v) = 2\}| = c_2$?
\end{enumerate}
  \end{itemize}\vspace*{-1ex}
\noindent\rule{\textwidth}{0.4pt}

In the following, we refer to $|\{v \in V \mid c(v) = 1\}|$ and $|\{v \in V \mid c(v) = 2\}|$ as color class~1~and~2, respectively. 

The rest of this paper is organized as follows: Section~\ref{sec2} presents the main result of this paper, a polynomial-time algorithm for $2$-MDCC as well as an engineered variant that shows significant runtime improvement in practice. In Section~\ref{sec:comp-results} we demonstrate the computational benefits of implementations of our algorithms over the previous integer linear programming solution. Section~\ref{sec:discussion} discusses possible improvements. We conclude with a short summary and outlook in Section~\ref{sec:conclusions}.

\section{Algorithms for $2$-MDCC}\label{sec2}
In this section we give polynomial-time algorithms for $2$-MDCC. We start by describing the base algorithm in Section~\ref{sec:base-algorithm} and present a practical improvement in Section~\ref{sec:algor-engin}.

\subsection{Base algorithm}
\label{sec:base-algorithm}

Our approach adopts the framework of the \textsc{OptDispF} algorithm from~\cite{papenbergExtendingBicriterionApproach2025a}, restricted to $k=2$ as outlined in Algorithm~\ref{alg1}. \textsc{OptDispF} is inspired by an algorithm for dispersion-based bipartitioning~\cite{Brucker1978}. At a high level, it iteratively increases a distance threshold and updates the edge set of a growing graph accordingly in each iteration. It then solves the $k$-COLCC subproblem in this graph using an integer linear programming (ILP) formulation. In this work, for the case $k = 2$, we replace the previously used ILP formulation with a dedicated combinatorial algorithm. This algorithm transforms the cardinality-constrained coloring problem into a bounded subset sum problem, which is solved by dynamic programming (DP).

\begin{algorithm}[htbp]
  \begin{algorithmic}[1]
\Statex \textbf{Input:} An instance $(I, d, c_1, c_2)$ of $2$-MDCC
\Statex \textbf{Output:} Optimal dispersion $\delta^*$ and anticluster assignment $\pi^*: I \rightarrow \{1, 2\}$ respecting the cardinality constraints
\Statex
\Procedure{OptDispF$^{\textsc{2-COLCC}}$}{$I, d, c_1, c_2$}
	\State compute dissimilarities $d(i,j)$ between all pairs $i, j \in I$ and store them in an array~$D$
	\State sort $D$ in ascending order
        \State initialize a graph $G = (I, E = \emptyset)$
	\State initialize an array $\pi$ with $|I|$ entries\Comment{assignment of items to anticluster 1 or 2}
	\For{\textbf{each} $\delta$ in $D$}
		\State add all edges $\{u, v\}$ to $E$ with $d(u, v) = \delta$
                \State $\pi' \gets$ \textsc{solve\_2-COLCC}($G$, $c_1$, $c_2$)
		\If{$\pi'$ is \textrm{false}} \textbf{break} \EndIf
		\State $\pi \gets \pi'$
    \EndFor
	\State \Return $\delta$, $\pi$ \Comment{$\delta$ is dispersion, $\pi$ is anticluster assignment}
\EndProcedure
\end{algorithmic}
\caption{OptDispF$^\text{2-COLCC}$ algorithm adapted from~\cite{papenbergExtendingBicriterionApproach2025a}.}
\label{alg1}
\end{algorithm}%

The key idea of our algorithm for cardinality-constrained 2-coloring is to exploit the structure of the bipartite components induced by the current graph of each iteration. Each connected component admits a bipartition, and the two parts can be assigned to the two color classes in either orientation. This reduces the problem to deciding, for each component, which side contributes to color class~1 and which to color class~2, subject to the cardinality constraints.

To simplify this decision, we do not track the sizes of both parts of the partition separately. Instead, we consider only their difference, which uniquely determines the effect of choosing one orientation over the other. Thus, each component contributes a single value encoding the imbalance between its two parts, and the problem reduces to deciding whether the larger part of the partition is assigned to color class~1 or not. Algorithm~\ref{alg2} gives the details.

\begin{algorithm}[H]
\begin{algorithmic}[1]
\Statex \textbf{Input:} An instance $(G, c_1, c_2)$ of $2$-COLCC 
\Statex \textbf{Output:} Either a coloring $c: V \rightarrow \{1, 2\}$ for each vertex respecting the cardinality constraints or \textrm{false} if no such coloring exists
\Statex
\Procedure{\textsc{solve\_2-COLCC}}{$G=(V,E)$, $c_1$, $c_2$}
	\If{$G$ is not bipartite} \Return \textrm{false} \EndIf
	\State $S \gets [\ ]$
	\State initialize coloring $c$ as an array with $|V|$ empty entries
	\State determine all $\ell$ connected components $C_1, C_2, C_3, \ldots, C_\ell$ in $G$
	\For{$i \gets 1, \ldots, \ell$}
        \State $P_i = Q_i = \emptyset$
    	\State \parbox[t]{\dimexpr\linewidth-\algorithmicindent-\algorithmicindent\relax}{perform a BFS on $C_i$ and assign the vertices discovered at even levels to $P_i$ and the vertices discovered at odd levels to $Q_i$}\vspace*{.25em}
		\State $S\textsc{.append}(\abs(|P_i| - |Q_i|))$\Comment{add size difference of $P_i$ and $Q_i$ to $S$}
    \EndFor
	\State $c_1' \gets c_1 - \sum_{i=1}^\ell \min\{|P_i|, |Q_i|\}$
	\State $c_2' \gets c_2 - \sum_{i=1}^\ell \min\{|P_i|, |Q_i|\}$
	\State $T \gets$ \textsc{SubsetSumTableFor2-COLCC}($S$)
	\If{$T[c_1']=\textrm{false}$} \Return \textrm{false}\EndIf
	\State{$j \gets c_1'$}
	\While{$j \neq 0$}
		\State $(i, s) \gets T[j]$
		\If{$|P_i| \geq |Q_i|$}
		\State color $P_i$ with Color 1 and $Q_i$ with Color 2, store colors in $c$ accordingly
		\Else{} \State color $P_i$ with Color 2 and $Q_i$ with Color 1, store colors in $c$ accordingly
		\EndIf
		\State $j \gets j - s$
	\EndWhile
	\State \Return $c$
\EndProcedure
\end{algorithmic}
\caption{2-COLCC algorithm. Lines 2--13 perform the feasibility check for the existence of a 2-coloring respecting the cardinality constraints, while Lines 14--21 implement the backtracking procedure that assigns the colors.}
\label{alg2}
\end{algorithm}%


The decision whether the larger part of a partition is assigned to either color class 1 or 2 is realized by a subset sum DP procedure (Algorithm~\ref{alg3}), where the input multiset consists of the imbalance values. While subset sum is NP-hard in general, here the standard pseudopolynomial DP runs in polynomial time because the total sum of all imbalance values is bounded by $n=|V|$. 
 Figure~\ref{fig1} shows an illustrative example of all steps of the $2$-COLCC algorithm.

\begin{figure}
\centering
\includegraphics[width=\textwidth]{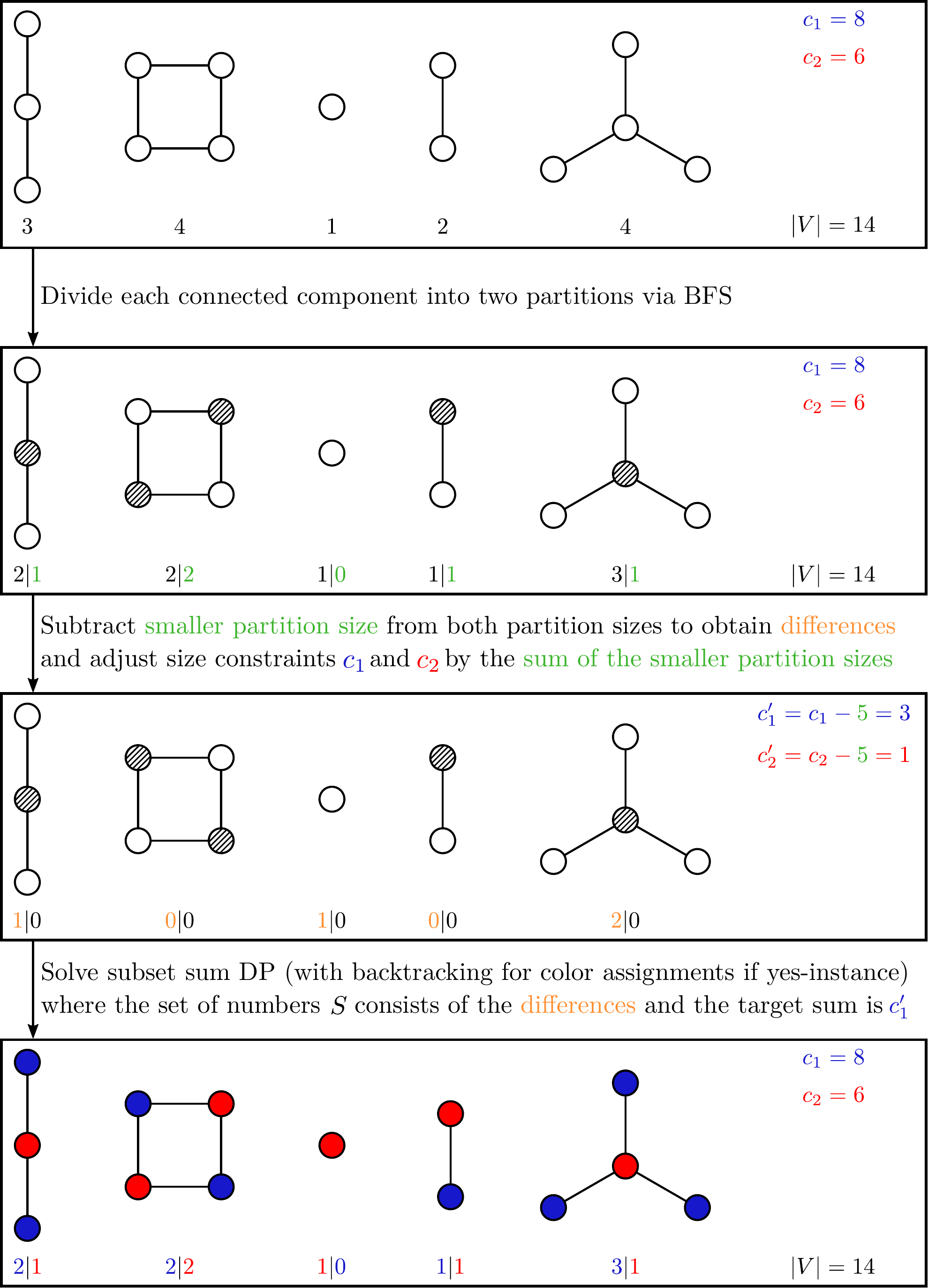}
\caption{2-COLCC algorithm. In the first panel, the number of vertices for each connected component is shown. Color~1 is blue, color~2 is red. A BFS divides each connected component into two partitions (plain and striped) that can be assigned different colors. Next, instead of the partition sizes themselves, only their differences are considered (yellow numbers). These are the input for the last step, which is a subset sum DP, with one of the adjusted cardinality constraints being the target value, where we can always choose $c_1'$. In this example, the yellow numbers of the left-most and right-most components were selected as the subset whose sum equals $c_1'$, resulting in the corresponding components having their larger partitions (plain) colored blue (as $c_1'$ is the target value), while the unchosen middle component has its larger partition (plain) colored red. The color assignment to the plain and striped partition for the other connected components is arbitrary because the partition sizes are equal.}
\label{fig1}
\end{figure}

We obtain the following main result:

\begin{ntheorem}
  Algorithm OptDispF$^{\textsc{2-COLCC}}$ computes an optimal solution of $2$-MDCC for $n$ items in time $\mathcal{O}(n^4)$ and space $\mathcal{O}(n^2)$. 
\end{ntheorem}

\begin{proof}
The main loop in Algorithm~\ref{alg1} is executed up to $\mathcal{O}(n^2)$ times while each execution of the constrained 2-coloring step takes time $\mathcal{O}(n^2)$ as the graphs can have a quadratic number of edges. The space requirement is dominated by storing the $\mathcal{O}(n^2)$ distances.
\end{proof}

\begin{algorithm}[H]
\begin{algorithmic}[1]
\Statex \textbf{Input:} A multiset of numbers $S$
\Statex \textbf{Output:} A table $T$ enabling backtracking the summands of all possible sums $[0,1,2,\ldots,\sum_{s \in S} s]$, where $T[x] = \textrm{false}$ for impossible sums $x$
\Statex
\Procedure{\textsc{SubsetSumTableFor2-COLCC}}{$S$, $c_1'$}
    \State $n \gets \sum_{s \in S} s$
	\State initialize DP table $T[0,1,2,\ldots,n]$ with \textrm{false}
	\State $T[0] \gets (0, 0)$
    \For{$i \gets 1,\ldots,|S|$}
		\State $s \gets S[i]$
		\For{$j \in [n, n-1, n-2, \ldots, s]$}
	        \If{$T[j-s]$ is not \textrm{false}}
				\State $T[j] \gets (i, s)$
			\EndIf
		\EndFor
    \EndFor
	\State \Return $T$
\EndProcedure
\end{algorithmic}
\caption{Standard subset sum DP for 2-COLCC.}
\label{alg3}
\end{algorithm}

\subsection{Algorithm Engineering}
\label{sec:algor-engin}

In our experiments, we use the Euclidean distance as the dissimilarity function between datapoints. The maximum dispersion is often attained within the first few iterations, likely because the underlying geometric structure of the Euclidean space typically tends to induce a large number of non 2-colorable triangles~\cite{BRINGMANN201935}. Generally, we observe an empirical upper bound of $|V|$ distances until an odd cycle forms. Of course, this is no guaranteed theoretical upper bound as the number of edges of a bipartite graph can be up to quadratic in the size of the vertex set. Nevertheless, we can use this observation to speed up the algorithm in practice.    

Since sorting all distances dominates the runtime of Algorithm~\ref{alg1} in our experiments, we focus on the $|V|$ smallest distances unless the maximum dispersion is not found among them. To this end, we compute distances in a streaming fashion to avoid storing the full distance matrix, while maintaining a max-heap of size $|V|$, into which new distances are inserted whenever they are smaller than the current maximum (Algorithm~\ref{alg4}). However, even if the maximum dispersion is found among the $|V|$ smallest distances, the asymptotic worst-case runtime still does not improve upon the $\mathcal{O}(n^2 \log n)$ worst-case runtime when sorting all distances, while---as demonstrated in the next section---the practical runtime improves and the space complexity is reduced from $\mathcal{O}(n^2)$ to $\mathcal{O}(n)$.

\begin{algorithm}[t]
  \begin{algorithmic}[1]
    \Statex \textbf{Input:} An instance $(I, d, c_1, c_2)$ of $2$-MDCC
\Statex \textbf{Output:} Optimal dispersion $\delta^*$ and anticluster assignment $\pi^*: I \rightarrow \{1, 2\}$ respecting the cardinality constraints
\Statex
\Procedure{OptDispF$^{\textsc{2-COLCC}}_{\textsc{HEAP}}$}{$I$, $d$, $c_1$, $c_2$}
	\State {initialize empty max-heap $D$ with capacity $|I|$}
	\State {\textbf{for each} distance $\delta$ in stream of distances between all $i \in I$ \textbf{do}}
	\State\hspace*{\algorithmicindent}{\textbf{if} $|D| = |I| \textbf{ and } \delta < D.\textsc{max}()$ \textbf{then}}
	\State \hspace*{\algorithmicindent}  \hspace*{\algorithmicindent}  { $D$.\textsc{extract-max()}}
	\State\hspace*{\algorithmicindent}\textbf{if} $|D| < |I|$ \textbf{then} {$D$.\textsc{insert($d$)}}
	\State {sort $D$ in ascending order with heap sort}
        \State initialize a graph $G = (I, E = \emptyset)$
	\State initialize an array $\pi$ with $|I|$ entries
	\For{\textbf{each} $\delta$ in $D$}
		\State add all edges $\{u, v\}$ to $E$ with $d(u, v) = \delta$
    	\State $\pi' \gets$ \textsc{solve\_2-COLCC}($G$, $c_1$, $c_2$)
		\If{$\pi'$ is \textrm{false}} \textbf{break} \EndIf
		\State $\pi \gets \pi'$
    \EndFor
	\If{$\pi'$ is not \textrm{false}} \Comment{if maximum dispersion not within $|I|$ smallest distances}
	\State {$\delta$, $\pi \gets$ \textsc{OptDispF$^\textsc{2-COLCC}$($I, d, c_1, c_2$)}}
	\EndIf
	\State \Return $\delta$, $\pi$ \Comment{$\delta$ is dispersion, $\pi$ is anticluster assignment}
\EndProcedure
\end{algorithmic}
\caption{\textsc{OptDispF$^\textsc{2-COLCC}_\textsc{HEAP}$} algorithm with the max-heap extension. Instead of sorting the distances, which takes up most of the time, only the $|V|$ smallest distances are regarded. If the maximum dispersion is not found within the $|V|$ smallest distances, \textsc{OptDispF$^\textsc{2-COLCC}$} is used afterwards.}
\label{alg4}
\end{algorithm}%
\newpage
\section{Computational Results}
\label{sec:comp-results}
Experiments were conducted on a workstation running Ubuntu~24.04~(64-bit). The system was equipped with an Intel~Core~i9-12900K~CPU and 31~GiB of RAM. The ILP solver used is Gurobi 13.0.1. All R and C++ code used in the
simulation study can be retrieved from the accompanying public GitHub repository
\url{https://github.com/AlBi-HHU/MDCC_2-COLCC}.

We varied the sample size $n$ in increments of $1{,}000$ up to $10{,}000$, while keeping the number of feature variables fixed at $m=2$ and the cardinality constraints at $c_1=\frac{n}{2}$ and $c_2=\frac{n}{2}$. All data were generated from a normal distribution with mean 0 and standard deviation 1.

The experiments comparing the ILP approach with the new algorithm without the max-heap optimization were repeated ten times to estimate average runtime and memory usage, where the new algorithm significantly outperforms the ILP approach in both metrics and exhibits more stable runtime behavior. See Figure~\ref{fig2}. Notably, the runtime is reduced from around a minute to a few seconds on average per run.

With the max-heap optimization, the runtime is even further reduced to well below a second, while the memory usage shows almost imperceptible linear growth rather than quadratic growth. See Figure~\ref{fig3}.

\begin{figure}[H]
\centering

    \begin{subfigure}{0.5\textwidth}
        \centering
        \includegraphics[width=\textwidth]{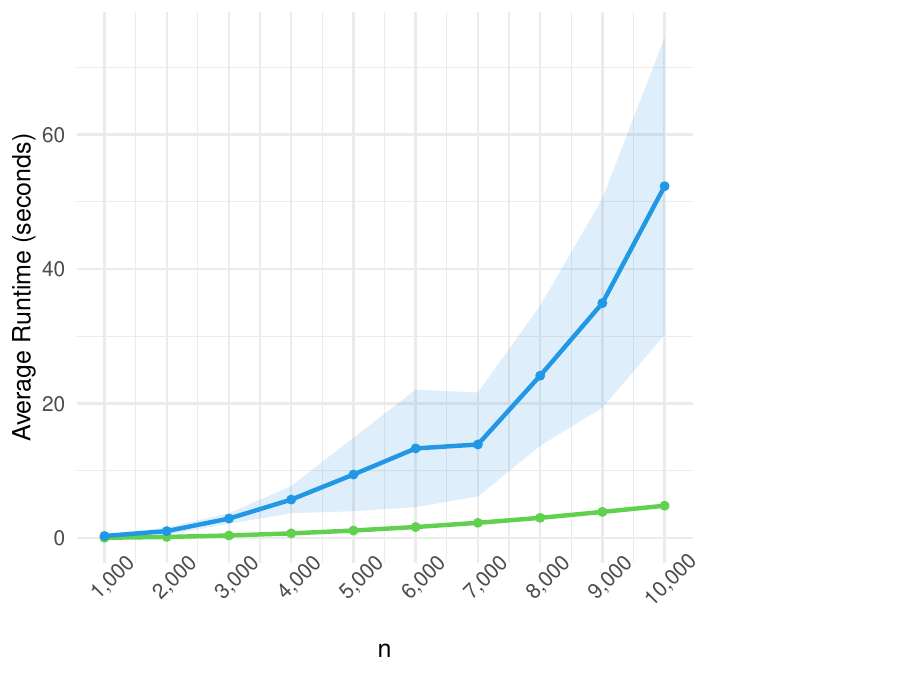}
    \end{subfigure}%
    \begin{subfigure}{0.5\textwidth}
        \centering
        \includegraphics[width=\textwidth]{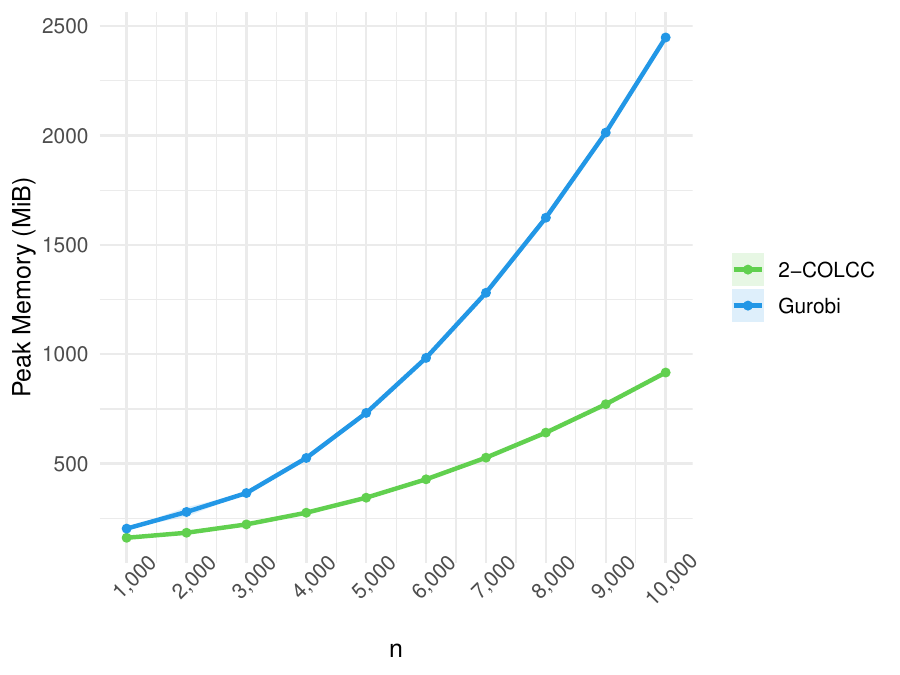}
    \end{subfigure}
\caption{Comparison of runtime and memory usage between the ILP solver (Gurobi) and the new approach. The shaded area indicates the standard deviation over 10 runs.}
\label{fig2}
\end{figure}

\begin{figure}[H]
\centering

    \begin{subfigure}{0.5\textwidth}
        \centering
        \includegraphics[width=\textwidth]{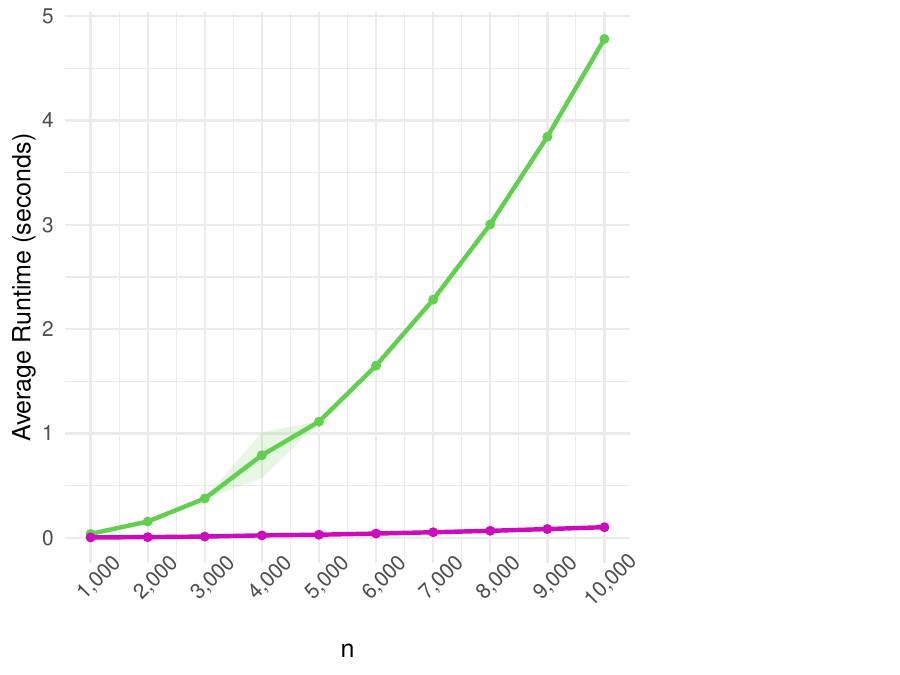}
    \end{subfigure}%
    \begin{subfigure}{0.5\textwidth}
        \centering
        \includegraphics[width=\textwidth]{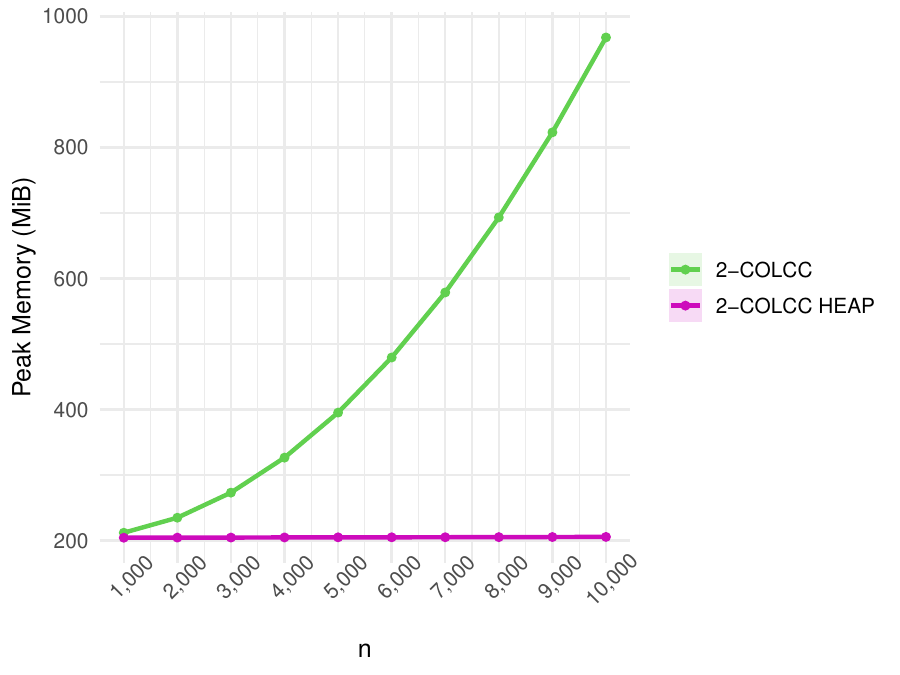}
    \end{subfigure}
\caption{Comparison of runtime and memory usage between \textsc{OptDispF}$^\text{2-COLCC}_\textsc{HEAP}$ and \textsc{OptDispF}$^\text{2-COLCC}$. The shaded area indicates the standard deviation over 10 runs.}
\label{fig3}
\end{figure}

\section{Discussion}\label{sec:discussion}


In theory, the $\mathcal{O}(n^2 \log n)$ runtime for selecting the $n$ smallest among all $\mathcal{O}(n^2)$ pairwise distances in Algorithm~\ref{alg4} could be improved by replacing the max-heap with a median-of-medians strategy. This strategy first identifies the $n$-th smallest distance using the median-of-medians algorithm~\cite{medianofmedians1973}, followed by a partitioning step to isolate the $n$ smallest distances, where both median-of-medians selection and partitioning can be performed in linear time with respect to the input size, i.e., the number of distances $\mathcal{O}(n^2)$ in our case. These $n$ smallest distances can then be sorted in $\mathcal{O}(n \log n)$ time, yielding a runtime complexity of $\mathcal{O}(n^2)$ for selecting the $n$ smallest distances, which improves upon the heap-based approach with a runtime of $\mathcal{O}(n^2 \log n)$.

However, our implementation of this approach did not yield performance improvements due to the large constant factors associated with the median-of-medians procedure. Moreover, the memory-efficient streaming formulation used in the heap-based variant requiring only $\mathcal{O}(n)$ space is not applicable here, as the median-of-medians approach requires the entire set of $\mathcal{O}(n^2)$ distances at the same time.

Interesting future research includes a polyhedral study of the linear programming relaxation of the ILP for the case $k = 2$ to check for total unimodularity or other interesting properties. We also find it an intriguing question to either refine the runtime analysis of the presented algorithm to obtain a tighter bound than $\mathcal{O}(n^4)$, because a graph constructed by adding edges in ascending order of Euclidean distance may never produce a bipartition with $\Theta(n^2)$ edges, or to come up with alternative, faster algorithms.

\section{Conclusions}\label{sec:conclusions}
We studied $k$-MDCC for $k=2$ and showed that the problem is solvable in polynomial time. This follows from the fact that the corresponding decision problem 2-COLCC is in P, as established in this work. Hence, $k$-MDCC is polynomial-time solvable in the case $k=2$, in contrast to the NP-hardness for $k \geq 3$.

This result may be useful in applications such as large-scale anticlustering problems, where exact solutions are required for large datasets. More broadly, it may be relevant for variants of related coloring and partitioning problems that arise in future applications, for which our approach suggests potential polynomial-time solvability in restricted settings.

\bibliography{2coloring_cardinality_constraints.bib}

\appendix
\renewcommand{\thefigure}{S\arabic{algorithm}}
\setcounter{figure}{0}

\end{document}